\documentclass[pra,aps,twocolumn,superscriptaddress]{revtex4}
\usepackage[colorlinks=true, allcolors=blue]{hyperref}
\usepackage{times}
\usepackage{verbatim}
\usepackage{comment}
\usepackage{amsmath,amssymb,amsthm,bm}
\usepackage{mathtools}

\usepackage{algorithm}
\usepackage[noend]{algpseudocode}
\usepackage{etoolbox}
\usepackage{xcolor}
\usepackage{tikz}
\usepackage[caption=false]{subfig}
\usepackage{float}
\newtheorem{theorem}{Theorem}
\newtheorem{definition}[theorem]{Definition}

\newtheorem{lemma}[theorem]{Lemma}

\newenvironment{mycases} {\left\{\begin{aligned}} {\end{aligned}\right\}}
\newcommand{\NN}{\mathbb{N}}
\newcommand{\brac}[1]{\left\lbrace #1   \right\rbrace}
\newcommand{\sbrac}[1]{\left[#1\right]}
\newcommand{\pbrac}[1]{\left(#1\right)}

\newcommand{\abs}[1]{\left\vert#1\right\vert}
\newcommand{\tr}[1]{\text{Tr}\brac{#1}}
\newcommand{\id}[1]{\mathbb{I}_{#1}}
\newcommand{\Arep}[1]{A_{#1}}

\newcommand{\mylhaf}{\mathrm{lhaf}_{\!A\mkern-1mu}}
\newcommand{\myPM}{\mathrm{PM}_{\mkern-1mu G\mkern-1mu}}

\DeclareMathOperator{\lhaf}{lhaf}

\DeclareMathOperator{\PM}{PM}
\DeclareMathOperator{\supp}{supp}
\DeclareMathOperator{\sat}{sat}
\definecolor{TK}{RGB}{200,10,10}
\definecolor{NQ}{RGB}{10,10,200}
\definecolor{DC}{RGB}{10,200,10}

\newcommand{\eq}[1]{\begin{align}#1\end{align}}

\newcommand{\braket}[1]{\langle #1 \rangle}
\begin{document}
\title{Efficient sampling from shallow Gaussian quantum-optical circuits with local interactions}
\author{Haoyu Qi}
\affiliation{Xanadu, Toronto, ON, M5G 2C8, Canada}
\author{Diego Cifuentes}
\affiliation{Department of Mathematics, Massachusetts Institute of Technology, Cambridge, MA, 02139-4307, USA}

\author{Kamil Br\'adler}
\affiliation{Xanadu, Toronto, ON, M5G 2C8, Canada}
\author{Robert Israel}
\affiliation{Xanadu, Toronto, ON, M5G 2C8, Canada}
\author{Timjan Kalajdzievski}
\affiliation{Xanadu, Toronto, ON, M5G 2C8, Canada}
\author{Nicol\'as Quesada}
\affiliation{Xanadu, Toronto, ON, M5G 2C8, Canada}
\begin{abstract}

We prove that a classical computer can efficiently sample from the photon-number probability distribution of a Gaussian state prepared by using an optical circuit that is shallow and local. Our work generalizes previous known results for qubits to the continuous-variable domain. The key to our proof is the observation that the adjacency matrices characterizing the Gaussian states generated by shallow and local circuits have small bandwidth. To exploit this structure, we devise fast algorithms to calculate loop hafnians of banded matrices.
Since sampling from deep optical circuits with exponential-scaling photon loss is classically simulable, our results pose a challenge to the feasibility of demonstrating quantum supremacy on photonic platforms with local interactions.
\end{abstract}
\maketitle
\section{Introduction}
Gaussian Boson Sampling (GBS) is a model of quantum computation in which a multimode Gaussian state is probed using photon-number resolving (PNR) detectors~ \cite{hamilton2017gaussian,kruse2019detailed,rahimi2015can}. Originally introduced as an experimentally friendly proposal to show
quantum computational supremacy~\cite{lund2014boson,barkhofen2017driven,aaronson2011computational}, it also finds potential applications~\cite{bromley2020applications} in chemistry~\cite{huh2015boson,jahangiri2020quantum}, optimization~\cite{arrazola2018using,arrazola2018quantum,banchi2020molecular}, graph theory~\cite{schuld2020measuring,bradler2018graph,bradler2018gaussian}, non-Gaussian state preparation~\cite{sabapathy2019production,su2019conversion}, and machine learning~\cite{banchi2020training}.  Experimentally, GBS is carried out by sending squeezed (and possibly displaced) single-mode states into an interferometer which mixes them. Typically, the interferometer is implemented by applying successive layers of beamsplitters that couple nearest-neighbour modes~\cite{reck1994experimental,clements2016optimal, de2018simple}. If the number of layers is large enough, linearly depending on the number of modes to be precise, one can implement arbitrary passive unitary transformations. 
It is precisely in this regime, when strong multi-partite entanglement is prevalent, that GBS is expected to be hard to simulate for a classical computer~\cite{hamilton2017gaussian,aaronson2011computational}. 

Noise and errors are inevitable in any experimental implementation of a near-term device. For photonic architectures, photon loss is the dominant source of error since it increases exponentially with the circuit depth. Therefore, it is not too surprising that a sampling device with a deep circuit loses its computational advantage asymptotically~\cite{garcia2019simulating,qi2020regimes}. Furthermore, faster classical algorithms~\cite{neville2017classical,clifford2018classical,quesada2019simulating,clifford2020faster} and large-scale simulation results~\cite{wu2018benchmark,gupt2020classical,li2020benchmarking} demand  larger and larger circuit size or input photon number. As a result, demonstrating quantum computational supremacy on current known photonic models has become increasingly difficult. 

Naturally, the next question to ask is: can we reduce the circuit depth to mitigate the adversarial effects of photon loss, but at the same time preserve the quantum advantage? In this paper we prove that the answer is no, if there are only local interactions. The classical simulability of qubit systems with shallow circuits and local gates was proved in Ref.~\cite{jozsa2006simulation} by representing these systems as matrix product states~\cite{vidal2003efficient}. Tensor network theory is a powerful tool to characterize and study the entanglement properties of quantum states generated by local interactions~\cite{montangero2018introduction}. However, for optical systems, tensor network simulation only leads to unsatisfying quasi-polynomial algorithms due to a non-constant physical dimension, a result of the photon-bunching effect~\cite{garcia2019simulating,qi2020regimes}.

Furthermore, using tensor network methods means we need to truncate the Hilbert space, which necessarily destroys the Gaussian structure. We circumvent this issue by coarse-graining the probabilities above a certain threshold. The merits of this step are threefold: 1) it eventually leads to a strictly polynomial runtime algorithm thanks to the constant threshold; 2) It preserves the Gaussian picture in the sense that our algorithm only requires calculating loop hafnians of the adjacency matrices; 3) The coarse-graining process is actually simulating the effect of the finite resolution of PNRs \cite{lita2008counting,levine2012algorithm,hadfield2009single}. 

In our simulation algorithm, the information of the circuit depth is captured by the bandwidth of the adjacency matrices. To exploit the banded structure, we introduce a new algorithm to calculate loop hafnians of banded matrices, extending previous results for permanents~\cite{cifuentes2016efficient}. We also devise an algorithm to exploit photon collisions, which are unavoidable for shallow circuits, leading to further speedup. We believe these methods are also of interest in the study of combinatorial graph problems and data structures~\cite{barvinok2016combinatorics}.

We organize our manuscript as follows. In Sec.~\ref{sec:GBS}, we provide a brief review of Gaussian states and their representation in the Fock basis by using loop hafnians. We then modify the GBS simulation algorithm from Ref.~\cite{quesada2020exact} in Sec.~\ref{sec:sampling} to incorporate the effect of the finite resolution of PNRs. It is crucial for any GBS simulation algorithm to have a well-defined run time, otherwise the photon-number distribution would have support on the (countably infinite) non-negative integers. In Sec.~\ref{sec:shallow-cicuit}, we study shallow circuit GBS with local interactions and show that the adjacency matrices arising from our sampling algorithm enjoy a banded structure. 
We also point out that photon collisions induce repetitions in the adjacency matrices. 
We exploit both of these special structures in Sec.~\ref{sec:haf-band}. Specifically, we prove that the time cost of calculating the loop hafnian of a banded matrix depends exponentially on its bandwidth. We also present an algorithm which calculates loop hafnians faster when there are repeated columns and rows.
Finally, in Sec.~\ref{sec:main}, we put all these results together to show that a classical computer can simulate GBS with local interactions in polynomial time, if the depth of the circuit is logarithmic in the number of modes. We discuss how our methods are also applicable to Boson Sampling and present the conclusion in Sec.~\ref{sec:conclusion}.

\section{Gaussian boson sampling with finite-resolution photon number detectors}
\label{sec:GBS}

An $M$-mode Gaussian state $\rho$ is fully characterized by its $2M$-dimensional  mean vector  and its $2M\times 2M$-dimensional covariance matrix, with entries:
\begin{align}
    \alpha_j &= \tr{\rho \hat{\alpha}_j},\\
    \sigma_{ij} & = \tfrac{1}{2}\tr{\rho \left( \hat{\alpha}_i \hat{\alpha}_j + \hat{\alpha}_j \hat{\alpha}_i \right)} - \alpha_i\alpha_j.
\end{align}
The vector $\hat{\alpha}:=(\hat{a}_1,\ldots,\hat{a}_M,\hat{a}^\dagger_1,\ldots,\hat{a}_M^\dagger)$ consists of ladder operators for each mode arranged in the shown order. We refer the readers to Refs.~\cite{weedbrook2012gaussian,serafini2017quantum} for a comprehensive review on Gaussian quantum information.

Due to the ease of preparing Gaussian quantum states, GBS was proposed as a candidate to demonstrate quantum computational supremacy~\cite{hamilton2017gaussian,kruse2019detailed}. Given a Gaussian state with vector of means $\alpha$ and covariance matrix $\sigma$, it can be shown that the probability of detecting photon pattern $s=(s_1,\ldots,s_M)$ is given by~\cite{quesada2019simulating}
\begin{align}\label{eq:prob-hafnian}
    p(s) &= \frac{\exp\pbrac{-\frac{1}{2}\alpha^\dagger Q^{-1} \alpha}}{\sqrt{\det(Q)}}\frac{\lhaf(\tilde{A}_s)}{s_1!\ldots s_M!}~,\\\label{eq:A-1}
    Q & := \sigma + \mathbb{I}_{2M}/2~,\\
    A & := \begin{pmatrix}
    0 & \mathbb{I}_M \\\label{eq:A-2}
    \mathbb{I}_M & 0
    \end{pmatrix} \left(\mathbb{I}_{2M}-Q^{-1}\right)~, \\ \label{eq:A-3}
    \tilde{A}_s& := \text{fdiag}(A_s, \gamma_s)~,\\
    \gamma^T& := \alpha^\dagger Q^{-1}.
\end{align}
Here $A=A^T$ is
the \textit{adjacency matrix} of the Gaussian state $\rho$. The displacement vector modifies the usual expression for the GBS probability~\cite{hamilton2017gaussian} as follows: 1) it adds an extra Gaussian prefactor which can be efficiently calculated, 2) it fills up the diagonal elements of the $A_s$ matrix: this is precisely what the function $\text{fdiag}(\cdot,\cdot)$ does, it replaces the diagonal entries of its first argument with the entries of its second argument.

Note that the output probability in Eq.~\eqref{eq:prob-hafnian} depends on the \textit{extended adjacency matrices} $\tilde{A}_s$, which is obtained by repeating the rows and columns of $A$ to obtain $A_s$ and then filling its diagonal with $\gamma_s$. Specifically, if $s_j$ photons are detected at the $j$-th mode, then the $j$-th and $(j+M)$-th row and column of $A$ are repeated $s_j$ times to obtain $A_s$. If $s_j=0$, we simply remove the $j$-th and $(j+M)$-th rows and columns. Similarly, $\gamma_s$ is obtained from $\gamma$ by repeating its $i$ and $i+M$  entries a total of $s_i$ times. 

The loop hafnian of a $d\times d$ symmetric matrix $A$ is defined as the number of perfect matchings of a weighted graph with loops that has $A$ as its adjancency matrix:
\begin{align}\label{lhaf-def}
    \lhaf(A) := \sum_{\pi\in \text{PM}(k)} \prod_{ij\in\pi} A_{ij}~,
\end{align}
where PM($k$) is the set of perfect matchings of a complete graph with loops. For a complete graph with $k$ vertices that has loops the number of perfect matchings $|\text{PM}(k)|$ is given by the $k^\text{th}$ telephone number \cite{bjorklund2019faster}. See Fig.~\ref{fig:single-pair} for a graphical description of the set PM($k=4$).

\begin{figure}[!t]
  \begin{subfloat}
  	
  	\begin{tikzpicture}[scale=.7, shorten >=1pt, auto, node distance=1cm, ultra thick]
  	\tikzstyle{node_style} = [circle,draw=black, inner sep=0pt, minimum size=4pt]
  	\tikzstyle{edge_style} = [-,draw=black, line width=2, thick, dashed]
  	\tikzstyle{edge_styleg} = [-,draw=orange , line width=2, ultra thick]
  	
  	\node[node_style] (v1) at (0,0) {};
  	\node[node_style] (v2) at (1,1) {};
  	\node[node_style] (v3) at (0,1) {};
  	\node[node_style] (v4) at (1,0) {};
  	
  	\draw[edge_styleg]  (v1) edge (v2);
  	\draw[edge_style]  (v1) edge (v3);
  	\draw[edge_style]  (v1) edge (v4);
  	\draw[edge_style]  (v2) edge (v3);
  	\draw[edge_style]  (v2) edge (v4);
  	\draw[edge_styleg]  (v3) edge (v4);
  	
  	\draw[edge_style]  (v1) to [loop left , in=120,out=240, looseness=10] (v1);
  	\draw[edge_style]  (v2) to [loop right, in=-60,out=60,looseness=10] (v2);
  	\draw[edge_style]  (v3) to [loop left , in=120,out=240, looseness=10] (v3);
  	\draw[edge_style]  (v4) to [loop right, in=-60,out=60,looseness=10] (v4);
  	
  	\end{tikzpicture}
  	\ 
  	\begin{tikzpicture}[scale=.7, shorten >=1pt, auto, node distance=1cm, ultra thick]
  	\tikzstyle{node_style} = [circle,draw=black, inner sep=0pt, minimum size=4pt]
  	\tikzstyle{edge_style} = [-,draw=black, line width=2, thick, dashed]
  	\tikzstyle{edge_styleg} = [-,draw=orange, line width=2, ultra thick]
  	
  	\node[node_style] (v1) at (0,0) {};
  	\node[node_style] (v2) at (1,1) {};
  	\node[node_style] (v3) at (0,1) {};
  	\node[node_style] (v4) at (1,0) {};
  	
  	\draw[edge_style]  (v1) edge (v2);
  	\draw[edge_styleg]  (v1) edge (v3);
  	\draw[edge_style]  (v1) edge (v4);
  	\draw[edge_style]  (v2) edge (v3);
  	\draw[edge_styleg]  (v2) edge (v4);
  	\draw[edge_style]  (v3) edge (v4);
  	
  	\draw[edge_style]  (v1) to [loop left , in=120,out=240, looseness=10] (v1);
  	\draw[edge_style]  (v2) to [loop right, in=-60,out=60,looseness=10] (v2);
  	\draw[edge_style]  (v3) to [loop left , in=120,out=240, looseness=10] (v3);
  	\draw[edge_style]  (v4) to [loop right, in=-60,out=60,looseness=10] (v4);
  	\end{tikzpicture}
  	\
  	\begin{tikzpicture}[scale=.7, shorten >=1pt, auto, node distance=1cm, ultra thick]
  	\tikzstyle{node_style} = [circle,draw=black, inner sep=0pt, minimum size=4pt]
  	\tikzstyle{edge_style} = [-,draw=black, line width=2, thick, dashed]
  	\tikzstyle{edge_styleg} = [-,draw=orange, line width=2, ultra thick]
  	
  	\node[node_style] (v1) at (0,0) {};
  	\node[node_style] (v2) at (1,1) {};
  	\node[node_style] (v3) at (0,1) {};
  	\node[node_style] (v4) at (1,0) {};
  	
  	\draw[edge_style]  (v1) edge (v2);
  	\draw[edge_style]  (v1) edge (v3);
  	\draw[edge_styleg]  (v1) edge (v4);
  	\draw[edge_styleg]  (v2) edge (v3);
  	\draw[edge_style]  (v2) edge (v4);
  	\draw[edge_style]  (v3) edge (v4);
  	
  	\draw[edge_style]  (v1) to [loop left , in=120,out=240, looseness=10] (v1);
  	\draw[edge_style]  (v2) to [loop right, in=-60,out=60,looseness=10] (v2);
  	\draw[edge_style]  (v3) to [loop left , in=120,out=240, looseness=10] (v3);
  	\draw[edge_style]  (v4) to [loop right, in=-60,out=60,looseness=10] (v4);
  	
  	\end{tikzpicture}
  	\
  	\begin{tikzpicture}[scale=.7, shorten >=1pt, auto, node distance=1cm, ultra thick]
  	\tikzstyle{node_style} = [circle,draw=black, inner sep=0pt, minimum size=4pt]
  	\tikzstyle{edge_style} = [-,draw=black, line width=2, thick, dashed]
  	\tikzstyle{edge_styleg} = [-,draw=orange, line width=2, ultra thick]
  	
  	\node[node_style] (v1) at (0,0) {};
  	\node[node_style] (v2) at (1,1) {};
  	\node[node_style] (v3) at (0,1) {};
  	\node[node_style] (v4) at (1,0) {};
  	
  	\draw[edge_style]  (v1) edge (v2);
  	\draw[edge_style]  (v1) edge (v3);
  	\draw[edge_style]  (v1) edge (v4);
  	\draw[edge_style]  (v2) edge (v3);
  	\draw[edge_style]  (v2) edge (v4);
  	\draw[edge_styleg]  (v3) edge (v4);
  	
  	\draw[edge_styleg]  (v1) to [loop left , in=120,out=240, looseness=10] (v1);
  	\draw[edge_styleg]  (v2) to [loop right, in=-60,out=60,looseness=10] (v2);
  	\draw[edge_style]  (v3) to [loop left , in=120,out=240, looseness=10] (v3);
  	\draw[edge_style]  (v4) to [loop right, in=-60,out=60,looseness=10] (v4);
  	\end{tikzpicture}
  	\
  	\begin{tikzpicture}[scale=.7, shorten >=1pt, auto, node distance=1cm, ultra thick]
  	\tikzstyle{node_style} = [circle,draw=black, inner sep=0pt, minimum size=4pt]
  	\tikzstyle{edge_style} = [-,draw=black, line width=2, thick, dashed]
  	\tikzstyle{edge_styleg} = [-,draw=orange, line width=2, ultra thick]
  	\node[node_style] (v1) at (0,0) {};
  	\node[node_style] (v2) at (1,1) {};
  	\node[node_style] (v3) at (0,1) {};
  	\node[node_style] (v4) at (1,0) {};
  	
  	\draw[edge_styleg]  (v1) edge (v2);
  	\draw[edge_style]  (v1) edge (v3);
  	\draw[edge_style]  (v1) edge (v4);
  	\draw[edge_style]  (v2) edge (v3);
  	\draw[edge_style]  (v2) edge (v4);
  	\draw[edge_style]  (v3) edge (v4);
  	
  	\draw[edge_style]  (v1) to [loop left , in=120,out=240, looseness=10] (v1);
  	\draw[edge_style]  (v2) to [loop right, in=-60,out=60,looseness=10] (v2);
  	\draw[edge_styleg]  (v3) to [loop left , in=120,out=240, looseness=10] (v3);
  	\draw[edge_styleg]  (v4) to [loop right, in=-60,out=60,looseness=10] (v4);
  	
  	\end{tikzpicture}
  	
  	\begin{tikzpicture}[scale=.7, shorten >=1pt, auto, node distance=1cm, ultra thick]
  	\tikzstyle{node_style} = [circle,draw=black, inner sep=0pt, minimum size=4pt]
  	\tikzstyle{edge_style} = [-,draw=black, line width=2, thick, dashed]
  	\tikzstyle{edge_styleg} = [-,draw=orange, line width=2, ultra thick]
  	\node[node_style] (v1) at (0,0) {};
  	\node[node_style] (v2) at (1,1) {};
  	\node[node_style] (v3) at (0,1) {};
  	\node[node_style] (v4) at (1,0) {};
  	
  	\draw[edge_style]  (v1) edge (v2);
  	\draw[edge_styleg]  (v1) edge (v3);
  	\draw[edge_style]  (v1) edge (v4);
  	\draw[edge_style]  (v2) edge (v3);
  	\draw[edge_style]  (v2) edge (v4);
  	\draw[edge_style]  (v3) edge (v4);
  	
  	\draw[edge_style]  (v1) to [loop left , in=120,out=240, looseness=10] (v1);
  	\draw[edge_styleg]  (v2) to [loop right, in=-60,out=60,looseness=10] (v2);
  	\draw[edge_style]  (v3) to [loop left , in=120,out=240, looseness=10] (v3);
  	\draw[edge_styleg]  (v4) to [loop right, in=-60,out=60,looseness=10] (v4);
  	
  	\end{tikzpicture}
  	\
  	\begin{tikzpicture}[scale=.7, shorten >=1pt, auto, node distance=1cm, ultra thick]
  	\tikzstyle{node_style} = [circle,draw=black, inner sep=0pt, minimum size=4pt]
  	\tikzstyle{edge_style} = [-,draw=black, line width=2, thick, dashed]
  	\tikzstyle{edge_styleg} = [-,draw=orange, line width=2, ultra thick]
  	\node[node_style] (v1) at (0,0) {};
  	\node[node_style] (v2) at (1,1) {};
  	\node[node_style] (v3) at (0,1) {};
  	\node[node_style] (v4) at (1,0) {};
  	
  	\draw[edge_style]  (v1) edge (v2);
  	\draw[edge_style]  (v1) edge (v3);
  	\draw[edge_style]  (v1) edge (v4);
  	\draw[edge_style]  (v2) edge (v3);
  	\draw[edge_styleg]  (v2) edge (v4);
  	\draw[edge_style]  (v3) edge (v4);
  	
  	\draw[edge_styleg]  (v1) to [loop left , in=120,out=240, looseness=10] (v1);
  	\draw[edge_style]  (v2) to [loop right, in=-60,out=60,looseness=10] (v2);
  	\draw[edge_styleg]  (v3) to [loop left , in=120,out=240, looseness=10] (v3);
  	\draw[edge_style]  (v4) to [loop right, in=-60,out=60,looseness=10] (v4);
  	\end{tikzpicture}
  	\
  	\begin{tikzpicture}[scale=.7, shorten >=1pt, auto, node distance=1cm, ultra thick]
  	\tikzstyle{node_style} = [circle,draw=black, inner sep=0pt, minimum size=4pt]
  	\tikzstyle{edge_style} = [-,draw=black, line width=2, thick, dashed]
  	\tikzstyle{edge_styleg} = [-,draw=orange, line width=2, ultra thick]
  	\node[node_style] (v1) at (0,0) {};
  	\node[node_style] (v2) at (1,1) {};
  	\node[node_style] (v3) at (0,1) {};
  	\node[node_style] (v4) at (1,0) {};
  	
  	\draw[edge_style]  (v1) edge (v2);
  	\draw[edge_style]  (v1) edge (v3);
  	\draw[edge_style]  (v1) edge (v4);
  	\draw[edge_styleg]  (v2) edge (v3);
  	\draw[edge_style]  (v2) edge (v4);
  	\draw[edge_style]  (v3) edge (v4);
  	
  	\draw[edge_styleg]  (v1) to [loop left , in=120,out=240, looseness=10] (v1);
  	\draw[edge_style]  (v2) to [loop right, in=-60,out=60,looseness=10] (v2);
  	\draw[edge_style]  (v3) to [loop left , in=120,out=240, looseness=10] (v3);
  	\draw[edge_styleg]  (v4) to [loop right, in=-60,out=60,looseness=10] (v4);
  	
  	\end{tikzpicture}
  	\
  	\begin{tikzpicture}[scale=.7, shorten >=1pt, auto, node distance=1cm, ultra thick]
  	\tikzstyle{node_style} = [circle,draw=black, inner sep=0pt, minimum size=4pt]
  	\tikzstyle{edge_style} = [-,draw=black, line width=2, thick, dashed]
  	\tikzstyle{edge_styleg} = [-,draw=orange, line width=2, ultra thick]
  	\node[node_style] (v1) at (0,0) {};
  	\node[node_style] (v2) at (1,1) {};
  	\node[node_style] (v3) at (0,1) {};
  	\node[node_style] (v4) at (1,0) {};
  	
  	\draw[edge_style]  (v1) edge (v2);
  	\draw[edge_style]  (v1) edge (v3);
  	\draw[edge_styleg]  (v1) edge (v4);
  	\draw[edge_style]  (v2) edge (v3);
  	\draw[edge_style]  (v2) edge (v4);
  	\draw[edge_style]  (v3) edge (v4);
  	
  	\draw[edge_style]  (v1) to [loop left , in=120,out=240, looseness=10] (v1);
  	\draw[edge_styleg]  (v2) to [loop right, in=-60,out=60,looseness=10] (v2);
  	\draw[edge_styleg]  (v3) to [loop left , in=120,out=240, looseness=10] (v3);
  	\draw[edge_style]  (v4) to [loop right, in=-60,out=60,looseness=10] (v4);
  	
  	\end{tikzpicture}
  	\
  	\begin{tikzpicture}[scale=.7, shorten >=1pt, auto, node distance=1cm, ultra thick]
  	\tikzstyle{node_style} = [circle,draw=black, inner sep=0pt, minimum size=4pt]
  	\tikzstyle{edge_style} = [-,draw=black, line width=2, thick, dashed]
  	\tikzstyle{edge_styleg} = [-,draw=orange, line width=2, ultra thick]
  	\node[node_style] (v1) at (0,0) {};
  	\node[node_style] (v2) at (1,1) {};
  	\node[node_style] (v3) at (0,1) {};
  	\node[node_style] (v4) at (1,0) {};
  	
  	\draw[edge_style]  (v1) edge (v2);
  	\draw[edge_style]  (v1) edge (v3);
  	\draw[edge_style]  (v1) edge (v4);
  	\draw[edge_style]  (v2) edge (v3);
  	\draw[edge_style]  (v2) edge (v4);
  	\draw[edge_style]  (v3) edge (v4);
  	
  	\draw[edge_styleg]  (v1) to [loop left , in=120,out=240, looseness=10] (v1);
  	\draw[edge_styleg]  (v2) to [loop right, in=-60,out=60,looseness=10] (v2);
  	\draw[edge_styleg]  (v3) to [loop left , in=120,out=240, looseness=10] (v3);
  	\draw[edge_styleg]  (v4) to [loop right, in=-60,out=60,looseness=10] (v4);
  	
  	\end{tikzpicture}
  \end{subfloat}
  
  \caption{\label{fig:single-pair}
  Perfect matchings including self-loops as defined in Eq.~\eqref{lhaf-def} for a complete graph with four vertices. The number of perfect matchings can also be obtained by using the single-variable signless matching polynomial. In this case, the corresponding matching polynomial is given by $\mu(z) = z^4 +6z^2 +3$~\cite{bradler2019duality}, and by setting $z=1$ we recover the same result as given by the loop hafnian calculation.}
\end{figure}
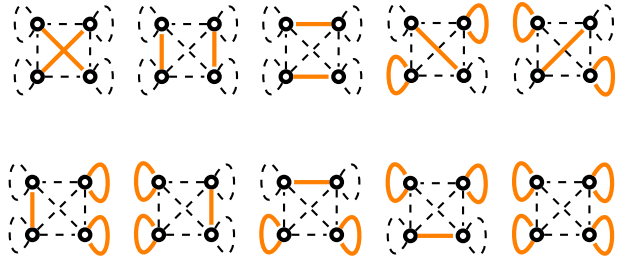

It might be worthwhile to point out that the loop hafnian is closely related to the \textit{multivariate (signless) matching polynomial} of the corresponding graph, a well-known quantity in graph theory~\cite{barvinok2016combinatorics}. The loop hafnian is equal to the matching
polynomial evaluated at a certain point~\cite{schuld2020measuring,bradler2019duality}. The matching polynomials could be a powerful tool since they enjoy a number of recursive relations~\cite{shi2016graph} and, in general, contain more
information about the graph than the loop hafnian.

The GBS task is to output samples according to the distribution $\brac{p(s)}$. Note that the support of this probability distribution is $\mathbb{N}^M$, i.e., the cartesian product of the non-negative integers $M$ times. A sampling problem over an infinite sample space can be ill-defined if we are aiming for a worst-case run-time analysis, since there is non-zero probability of detecting an arbitrary high number of photons. 

One might attempt to redefine the computational task as sampling over a post-selected distribution with fixed total photon number. Although such modification renders the sampling task well-defined, devising a classical algorithm to simulate such post-selected distribution might be challenging for general GBS scenarios~\cite{wu2020speedup}. This is because a system with fixed number of photons is described naturally by the particle representation, while Gaussian systems, with indefinite total photon number, are formulated under the mode representation.

However, such an issue should not be relevant for any experimental implementation, since a realistic device with finite energy cannot probe the system with infinite precision, nor can it detect an infinite amount of photons. Any realistic photon-number detector only has finite resolution power, i.e., it can only distinguish incoming photons up to a certain \emph{finite} number \cite{lita2008counting,levine2012algorithm,hadfield2009single}. We say that a detector is overloaded if the number of incoming photons is actually beyond the resolution of the detector. The special case where the detectors are overloaded by one or more photons has been studied in detail in Ref.~\cite{quesada2018gaussian}.

Therefore, to circumvent the divergence problem, we propose to modify the sampling task to incorporate the finite resolution of the detectors. Specifically, we modify the distribution by coarse-graining all photon patterns that overload at least one of the detectors. A two-mode example is shown in Fig.~\ref{fig:dist-2mode}.
Formally, we define Gaussian Boson Threshold Sampling (GBTS), as follows:
\begin{definition}[GBTS]
Consider an $M$-mode Gaussian state probed with PNR detectors with resolution $c$. Denote $\Sigma_c$ as the set of photon patterns which overload at least one PNR:
\begin{align}
    \Sigma_c &= \bigcup_{j=c}^\infty\Sigma_j~,\\
    \Sigma_j &= \brac{{s}:\exists j~ s_j>c}~.
\end{align}
The GBTS computational task is to output a sample from the following distribution:
\begin{align}\label{eq:GBTS-dist}
        \tilde{p}(x)=
        \begin{cases}
        p(s), \quad x={s}\notin \Sigma_c~,\\
        \sum_{s\in\Sigma_c}p(s), \quad x=\#~.
        \end{cases}
\end{align}
Here $p(s)$ is the output probability of the corresponding ideal GBS model given in Eq.~\eqref{eq:prob-hafnian}. We use the symbol $\#$ to indicate that at least one of the detectors is overloaded.
\end{definition}
Note that one can obtain bounds for the probability of occurrence of the $\#$ event in polynomial time by using marginal probabilities as shown in Appendix B of Ref.~\cite{quesada2020exact}.

\begin{figure}[!t]
    \centering
    \includegraphics[scale=0.6]{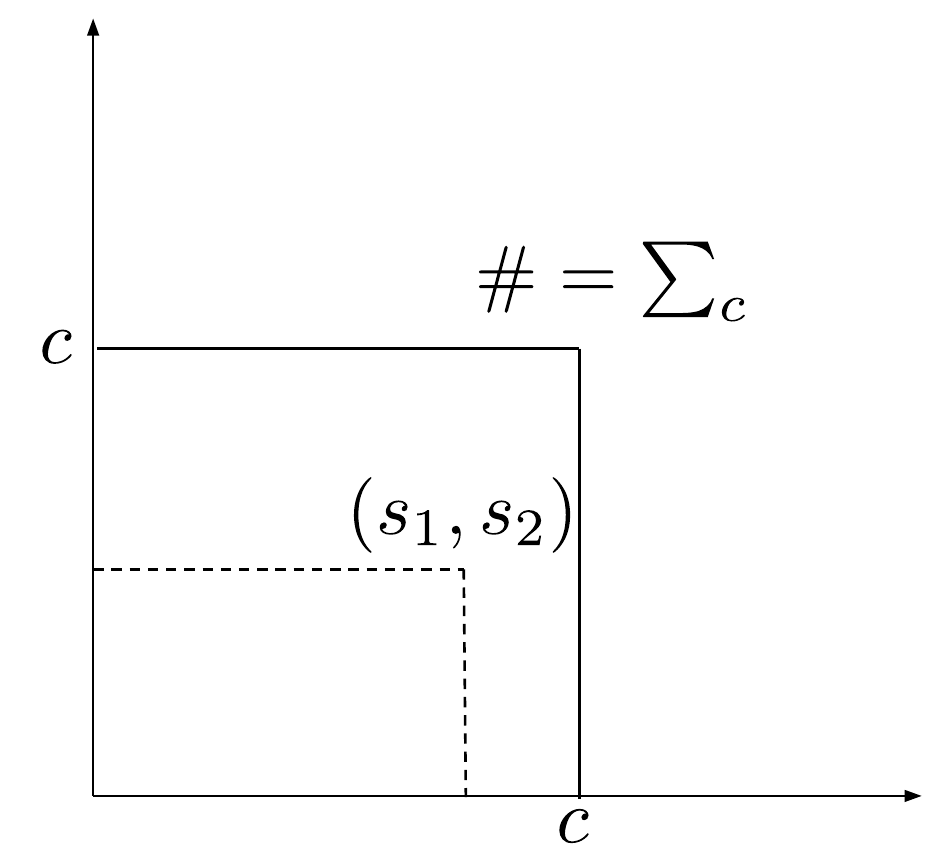}
    \caption{A two-mode example of the distribution of GBTS. We coarse-grain all outcomes outside  the `box' $[0,c] \times [0,c]$ into one event labelled by $\#$.}
    \label{fig:dist-2mode}
\end{figure}

\section{Simulating Gaussian boson sampling by calculating a polynomial number of probabilities }
\label{sec:sampling}

Naively, one can simulate any quantum device by calculating exponentially many probabilities. However, since the hardness of sampling originates from the hardness of calculating one probability~\cite{aaronson2011computational,dalzell2020many}, such a brute-force algorithm is expected to be far from optimal. Indeed, several fast algorithms to simulate Boson Sampling have been proposed, which output a sample by only calculating a polynomial number of probabilities \cite{neville2017classical,clifford2018classical,clifford2020faster}. For GBS, the authors of Ref.~\cite{quesada2019simulating} devised an algorithm which fully exploits the Gaussian nature of the system, in particular the fact that the reduced Gaussian states can be easily obtained. With a small modification, that algorithm can be used to sample exactly from the distribution of our GBTS problem.

The essence of the algorithm is to break down the original sampling problem into a chain of smaller GBS problems. At the $k$-th step, we sample the number of photons detected at the $k$-th mode, conditioned on the number of photons already sampled at previous steps. Specifically, at the $k$-th step, we sample from the following distribution:
\begin{align}
    q^{(k)}(x) = \begin{cases}
    p(x\vert s_1\ldots s_{k-1})~,\text{ for} ~~x=0,1,\ldots, c~.\\
    1 - \sum_{i\leq c}q^{(k)}(i)~,\text{  for} ~~x  = \text{`}>\text{'}~.
    \end{cases}
\end{align}
Here the outcome labelled by `$>$' represents all events that overload the $k$-th photon detector. If such outcome is produced, we simply output `$\#$' and exit the algorithm. This procedure is outlined in Algorithm~\ref{alg:GBTS}. The correctness of our algorithm, i.e., that it indeed samples according to the distribution given in Eq.~\eqref{eq:GBTS-dist}, is shown in Appendix~\ref{sec:append-correctness-algorithm}.

\begin{algorithm}[!t]
\caption{Simulating GBTS}\label{alg:GBTS}
\begin{algorithmic}[1]
\Procedure{GTBS}{}
\State $s\gets$ Empty array of length $M$
\For{$k\in [1,M]$}\Comment{Sample the number of photons in the $k$-th mode}
    \State $q^{(k)}\gets$ Empty array of length $c+2$
    \For{$x \in [0,\ldots, c,>]$}\Comment{Store the conditional distribution}
        \If{$i\leq c+1$}
            \State $q^{(k)}[i] \gets p(i\vert s[1]\ldots s[k-1])$\newcommand*\mean[1]{\bar{#1}}
        \Else
            \State $q^{(k)}[i]\gets 1-\sum_{i\leq c}q^{(k)}[i]$
        \EndIf
    \EndFor
    \State $s[k]\gets$\textsc{Sample}$(q^{(k)})$
    \If{$s[k]==c+1$}\Comment{Overloaded}
        \State\textbf{return} '$\#$'\Comment{so we coarse-grain it into event $\#$}
    \EndIf
\EndFor
\State \textbf{return} $s$
\EndProcedure
\end{algorithmic}
\end{algorithm}

We calculate the conditional probabilities by writing them as ratios of marginal probabilities:
\begin{align}
    p(x\vert s_1,\ldots,s_{k-1}) = \frac{p(s_1,\ldots,s_k)}{p(s_1,\ldots,s_{k-1})}~, x \leq c~.
\end{align}
The denominator comes for free since it must be already calculated at the previous step. The numerator is nothing but the output probability of the reduced Gaussian state on the first $k$ modes: 
\begin{align}
p(s_1,\ldots,s_k) =& \frac{\exp\pbrac{-\frac{1}{2}\alpha^{(k) \dagger} \{Q^{(k)}\}^{-1} \alpha^{(k)}}}{\sqrt{\det(Q^{(k)})}} \\
& \times \nonumber 
\frac{\lhaf\left[\left( \tilde{A}^{(k)} \right)_{ s_1,\ldots,s_k} \right]}{s_1!\ldots s_k!}~.
\end{align}
We use superscript $(k)$ to denote quantities associated with the reduced state of the first $k$-modes. As mentioned before, calculating the mean vector (covariance matrix) of the reduced state is particularly simple for a Gaussian system: to trace off the $j$-th mode, we simply remove the $j$-th and $j+M$-th rows (and columns) of the mean vector $\alpha$ and the covariance matrix $\sigma$ (or $Q$).

This concludes the presentation of our sampling algorithm. We illustrate the basic idea behind the algorithm in Fig.~\ref{fig:algorithm}.
It is not difficult to see that our algorithm only needs to calculate $Mc$ probabilities to output one sample. This is an exponential speed-up compared to the algorithms calculating all $c^M$ probabilities.
\begin{figure}[!t]
    \centering
    \includegraphics[width=0.48\textwidth]{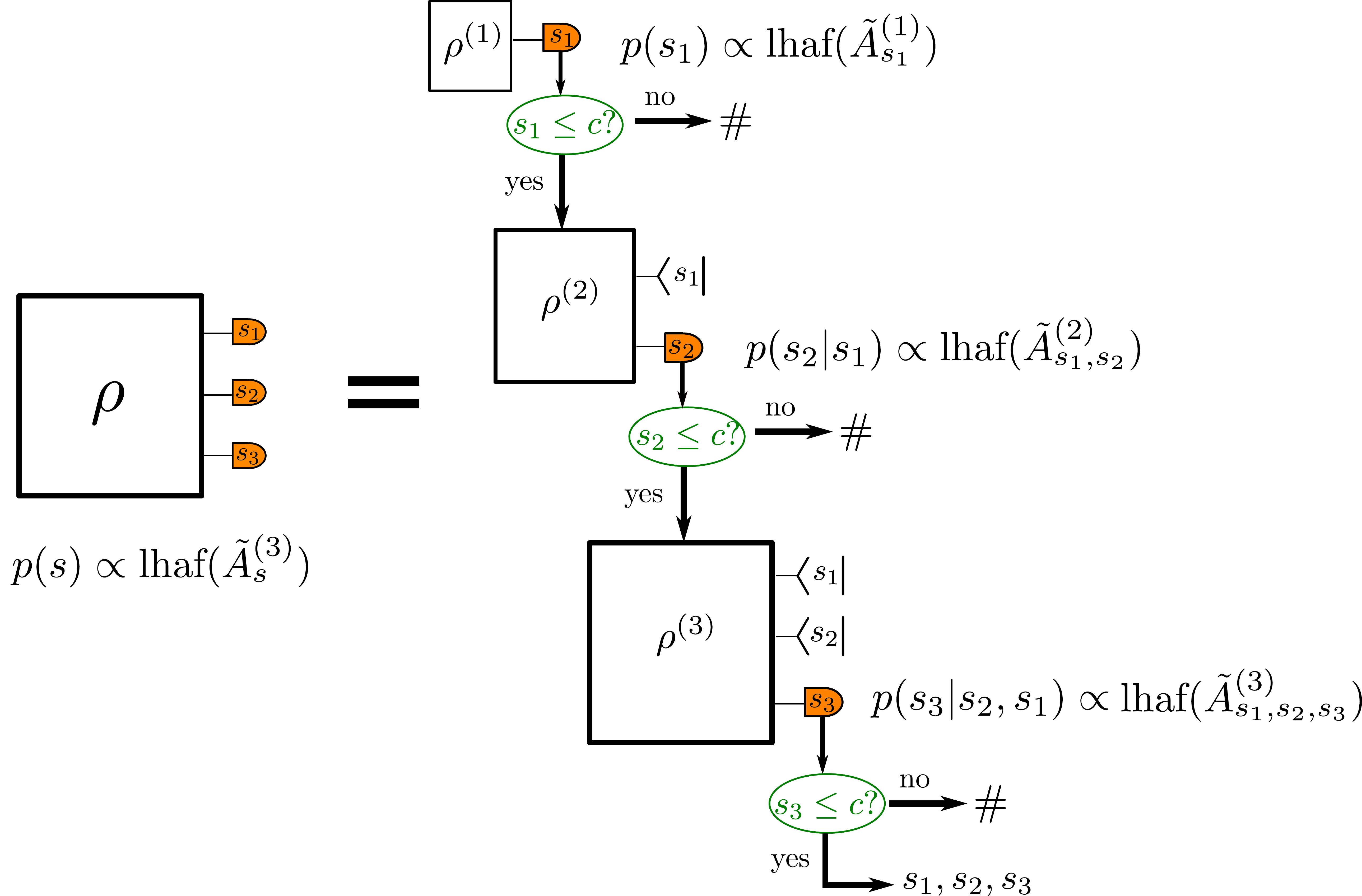}
    \caption{Block diagram of our algorithm for 3 modes. We explicitly show the control statements incorporating the possibility of overloading any detector.
    \label{fig:algorithm}}
\end{figure}
To explicitly write down the run time of this algorithm, we need to know how fast can we calculate each probability (loop hafnian). To answer this question, we need to first study the structures of the (extended) adjacency matrices appearing at each step of our algorithm.

\section{Adjacency matrices of Shallow circuits with local gates}
\label{sec:shallow-cicuit}
Formally, an $M$-mode optical circuit is called shallow if the number of layers of optical elements scales as $D=\log(M)$. Here, we count a group of commuting gates as one layer. We say that a two-mode gate or an optical element is local, if it acts on two adjacent modes.

Intuitively, a shallow circuit generates less quantum entanglement compared to a deep one: we expect that such system is easy to simulate. It turns out this is true, and the key to prove it lies in the observation that the relevant adjacency matrices involved in sampling from such a shallow circuit have banded structures which we can exploit to speed up the simulation.

\begin{definition}
A matrix $A$ is \textit{banded} with bandwidth $w$ if
$A_{i,j}=0$ for all $\abs{i-j}>w$. 

\end{definition}

We would like to show that, for GBS with a shallow circuit of local gates, the extended adjacency matrices of each reduced Gaussian state are banded matrices. We first prove that the unitary matrix is banded.

\begin{lemma}
\label{lem:U-band}Consider an $M$-mode optical circuit with $D$ layers of local gates, the associated unitary transformation is a banded matrix with bandwidth $D$.
\end{lemma}
\begin{proof}
Observe that one layer of local (two-mode and single-mode) gates is represented by a unitary matrix with bandwidth equal to $1$. This is because each gate at most entangles two adjacent modes. Since the bandwidth is additive under matrix multiplication, repeatedly applying this property, we end up with a banded unitary matrix with bandwidth $D$.
\end{proof}

\begin{lemma}
\label{lemma:dragonfly} Consider an $M$-mode GBS circuit with D layers of local gates and uniform loss.
The adjacency matrix of the first $k$, $1\leq k\leq M$, modes has the form
\eq{\label{eq:A-lossy-2}
A^{(k)} & =\begin{pmatrix}B^{(k)} & C^{(k)}\\
(C^{(k)} )^T & (B^{(k)})^{*}
\end{pmatrix},
}
where $B^{(k)}$ and  $C^{(k)}$  have bandwidth $w \leq 4D$.
\end{lemma}
\begin{proof}
We will show that the inverse of $Q^{(k)}$,  which is the covariance matrix of the first $k$ modes, is block banded. If this is the case, it trivially follows that $A^{(k)}$ is block banded as it is given by 
\begin{align}
A^{(k)} & := \begin{pmatrix}
0 & \mathbb{I}_k \\
\mathbb{I}_k & 0
\end{pmatrix} \left[\mathbb{I}_{2k}-\left(Q^{(k)} \right)^{-1}\right].
\end{align}

Before starting the proof, we set up some basic notation. We will consider a set of $M$ mixed single-mode states which are characterized by the following moments
\begin{align}
n_i = \braket{ \delta a^\dagger_i \delta a_i}  \geq 0, \quad 
m_i = \braket{\delta a_i^2}=\braket{\delta a_i^{\dagger 2}}^*, 
\end{align}
where we defined $\delta a_i = a_i - \braket{a_i}$. 
Pure squeezed states with squeezing parameters $r_i$ going through a circuit with uniform losses by overall energy transmission $\eta$ can be accommodated in the parametrization introduced above by setting
\begin{align}
n_i = \eta \sinh^2 r_i, \quad m_i = \tfrac12    \eta \sinh 2r_i,
\end{align}
where the lossless case is recovered by setting $\eta = 1$.
 
Since we are interested in local circuits with depth $D$, the matrix $U$, describing the interferometer mixing the $M$ single-mode states, is banded with bandwidth $D$ as proved in Lemma~\ref{lem:U-band}. 
With this notation we write
\begin{align}
Q =&\; \frac{\mathbb{I}_{2M}}{2} + \sigma, \quad   \sigma = V
T V^\dagger, \\ 
V =&\, \begin{pmatrix}
U^* & 0 \\
0 & U
\end{pmatrix}, \\
T =&\, \begin{pmatrix}
\bigoplus_{i=1}^M\ (n_i+\tfrac12)  & \bigoplus_{i=1}^M \ m_i \\
\bigoplus_{i=1}^M \ m_i^* & \bigoplus_{i=1}^M \ (n_i+\tfrac12)
\end{pmatrix}.
\end{align}

We can write the adjacency matrix of the $M$ modes as \cite{jahangiri2020point,rahimi2015can}
\begin{align}
A =&\; X (\mathbb{I}_{2M}-Q^{-1}) = \begin{pmatrix}
B & C \\
C^T & B^*
\end{pmatrix},\\	
B =&\; U \left( \bigoplus_{i=1}^M  \lambda_i   \right) U^T = B^T,\\
C =&\; U \left( \bigoplus_{i=1}^M \mu_i \right) U^\dagger = C^\dagger, \\
\lambda_i =&\; \frac{m_i}{(1+n_i)^2-|m_i|^2},\\
\mu_i =&\; 1 - \frac{1+n_i}{(1+n_i)^2-|m_i|^2}.
\end{align}
For pure states, one finds that $\mu_i = 0$ for all $i$. More generally, it holds that $B$ and $C$ are banded with bandwidth $2D$.

To proceed, we first note that
\begin{align}\label{eq:formW}
Q^{(k)} =& \frac{\mathbb{I}_{2k}}{2} + W_k \sigma  W_k ^\dagger,\\
W_k :=& \begin{pmatrix}
E_k & 0 \\
0 & E_k
\end{pmatrix},\\
E_{k} :=& \begin{pmatrix}
\id{k} & 0
\end{pmatrix}.
\end{align}
In the last equation $E_k$ is a $k \times M$ matrix,
and thus $W_k$ has size $2k \times 2M $.
Note that Eq.~\eqref{eq:formW} has precisely the form of the left hand side of the  Sherman-Morrison-Woodbury identity \cite{hager1989updating}
\begin{align}
\left(A + V C U \right)^{-1} = A^{-1} - A^{-1} V \left(C^{-1} + U A^{-1} V \right)^{-1} U A^{-1},
\end{align}
and thus we write
\begin{align}
\left[ Q^{(k)} \right]^{-1} =&\;
2 \mathbb{I}_{2 k}
- 4 W_k S^{-1}
W_k ^\dagger,
\\
S \,:=&\; \sigma^{-1} + 2 W_k ^\dagger W_k.
\end{align}
We claim that the matrix $S^{-1}$ is block banded with bandwidth at most $4D$.
Proving this claim will conclude the proof,
as it then follows that $[Q^{(k)}]^{-1}$ is also block banded, with the same bandwidth.

Now we examine the terms in the matrix~$S$.
Note that 
\begin{align}
W_k ^\dagger W_k &=
\begin{pmatrix}
E_k^\dagger E_k & 0 \\
0 & E_k^\dagger E_k
\end{pmatrix},
\\
E_k^\dagger E_k &=
\begin{pmatrix}
\id{k} & 0 \\
0 & 0
\end{pmatrix}.
\end{align}
Next we look at 
\begin{align}
\sigma^{-1} &= V T^{-1} V^\dagger,  \\
T^{-1} &=  \left(
\begin{array}{cc}
\bigoplus_{i=1}^M  \frac{n_i+\tfrac{1}{2}}{\left(n_i+\tfrac{1}{2}\right)^2 - |m_i|^2}  &
\bigoplus_{i=1}^M  \frac{-m_i}{\left(n_i+\tfrac{1}{2}\right)^2 - |m_i|^2}  \\
\bigoplus_{i=1}^M   \frac{-m_i^*}{\left(n_i+\tfrac{1}{2}\right)^2 - |m_i|^2}  &
\bigoplus_{i=1}^M \frac{n_i+\tfrac{1}{2}}{\left(n_i+\tfrac{1}{2}\right)^2 - |m_i|^2} \\
\end{array}
\right).
\end{align}
We are now ready to write
\begin{align}\label{eq:big}
S^{-1} =&\;  V \tilde{S}^{-1} V^\dagger,
\\
\tilde{S} \,:=&\;
T^{-1} +  2 V^\dagger  \,(W_k ^\dagger W_k)\, V.
\end{align}
Recall that the matrix $U$ describing the interferometer
is the product of precisely $D$ unitary matrices that are block diagonal,
where the blocks have either size one or two.
Similarly, the matrix $V = U^* \oplus U$ is also a product of block diagonal unitary matrices.
Given the special block structure of $W_k ^\dagger W_k$,
we conclude that
\begin{align}
V^\dagger
\,(W_k ^\dagger W_k)\,
V  =
\begin{pmatrix}
Y & 0 \\
0 & Y
\end{pmatrix}.
\end{align}
where $Y = \mathbb{I}_k \oplus K \oplus 0_{M-2D-k}$
and $K$ is a $2D \times 2D$ positive semi-definite matrix.
Since the blocks of $T^{-1}$ are diagonal,
it follows that the inverse of $\tilde{S}$ has the form
\begin{align}
\tilde{S}^{-1}
= \begin{pmatrix}
G & F \\
F^T & G^* 
\end{pmatrix},
\end{align}
where $G$ and $F$ are diagonal except for a block of size $2D \times 2D$.
After multiplying by $V$ on the left and by $V^\dagger$ on the right,
we conclude that $S^{-1}$ is block banded with bandwidth at most~$4D$, as claimed.
\end{proof}

We can improve our upper bound on the bandwidth to $w\leq 2D$. This can be understood by noting that each layer of local gates actually has `half' depth, so we have depth one for two layers of gates. Nonetheless, the upper bound of $4D$ is sufficient for our purposes.

\noindent We then observe the following:
\begin{lemma}\label{lemma:permutation}
There exists a permutation of rows and columns which transforms a block banded matrix $A^{(k)}$, where each block has bandwidth $w$, into a banded matrix with bandwidth~$2w$.
\end{lemma}
\begin{proof}
The permutation is given by 
\eq{\label{eq:permutation}
(1,\ldots,k,k+1,\ldots,2k) \rightarrow (1,k+1,\ldots,k,2k)~.
}
\end{proof}

Since diagonal elements do not affect the bandwidth, our results apply to $\tilde{A}^{(k)}$ as well. Combining Lemma~\ref{lemma:dragonfly} and Lemma~\ref{lemma:permutation}, it follows that, for GBTS with threshold $c$, at each step of Algorithm~\ref{alg:GBTS} we need to calculate hafnians of a banded matrix with bandwidth at most $8Dc$.

\section{Fast computation of loop hafnians for banded matrices}
\label{sec:haf-band}
The hafnian and loop hafnian are generalizations of the permanent,
which is $\sharp$P-complete to compute in the general setting~\cite{valiant1979complexity}.
The best known algorithm for computing permanents of arbitrary matrices was introduced by Ryser~\cite{ryser1963combinatorial} and has complexity $O(n\, 2^{n})$.
As for hafnians,
Bj{\"o}rklund~\cite{bjorklund2012counting} and Cygan and Pilipczuk~\cite{cygan2015faster} derived algorithms with complexity
$O(\mathrm{poly}(n)\, 2^{n/2})$ over arbitrary rings.
Subsequent work by Bj{\"o}rklund \emph{et al.}~\cite{bjorklund2019faster}
computed loop hafnians of complex matrices in time $O(n^3\, 2^{n/2})$.

The goal of this section is to derive an algorithm for efficiently computing loop hafnians of banded matrices,
in time ${O}(n\, w\, 4^{w})$ where $w$ is the bandwidth.
Previous work by Cifuentes and Parrilo~\cite{cifuentes2016efficient} proved that the computation of permanents of a banded matrix can be done in time ${O}(n\, w^2\, 4^{w})$.
Similarly, Schwartz~\cite{schwartz2009efficiently} gave an $O(8^w \log n)$
algorithm for computing hafnians of matrices which are both banded and Toeplitz. Finally, Temme and Wocjan~\cite{temme2012efficient} provided efficient algorithms for computing permanents of matrices that are block factorizable.

\begin{theorem}
	\label{thm:hafnian-band}
	Let $A$ be a symmetric $n\times n$ matrix with bandwidth~$w$ over an arbitrary ring.
	Then we can compute its loop hafnian using ${O}(n\, w\, 4^{w})$ arithmetic operations.
\end{theorem}

We first introduce some notation.
Let $G = (V,E)$ be the underlying graph structure of a banded matrix.
The vertex set is $V = [n] = \{1,2,3,\ldots,n \}$,
and the edge set $E$ consists of all pairs $(i,j)$ with $i \leq j \leq i{+}w$.
Given a list of edges $\pi \subset E$, we denote
\begin{align}
A(\pi) := \prod_{ij \in \pi} A_{ij}.
\end{align}
With the above notation, we have that
\begin{align}
\lhaf(A) = \sum_{\pi \in \PM(G)} A(\pi),
\end{align}
where $\PM(G)$ is the set of perfect matchings of~$G$.

Throughout this section we assume that the matrix $A$ is fixed.
Given a subset of indices $D \subset [n]$
we denote $\myPM(D)$ the set of perfect matchings of the restricted graph~$G|_D$.
We consider the \emph{subhafnian}
\begin{align}
\mylhaf(D) &:= \sum_{\pi \in \myPM(D)} A(\pi).
\end{align}
Equivalently $\mylhaf(D)$ is the loop hafnian of the principal submatrix of~$A$ indexed by~$D$.

Our algorithm to compute $\lhaf(A)$ relies on dynamic programming,
and is presented as Algorithm~\ref{alg:haf}.
For each $t \in [n]$ we compute a table $H_{t}$.
The loop hafnian of $A$ is one of the entries in the last table~$H_n$. See Appendix~\ref{apendix-example} for an explicit example where we present how to apply our algorithm to calculate the loop hafnian of a $5\times 5$ matrix.

\begin{algorithm}
	\caption{Banded Loop Hafnian}
	\label{alg:haf}
	\begin{algorithmic}[1]
		\Procedure{LHafBand}{$A,w$}
		\State $X_1 \gets \{1\}$
		\State $H_1(\bar D) \gets
		\begin{mycases}
		&1,
		&&\text{if } \bar D \!=\! \emptyset
		\\[-3pt]
		&A_{1,1},
		&&\text{if } \bar D \!=\! \{1\}
		\end{mycases}
		\;\text{\algorithmicforall\ } \bar D \subset X_1
		$
		\For {$t \in [2,n]$}
		\State $a_{t} \gets \max\{t{-}2w, 1\}$,\; $X_{t} \gets \{a_{t},\dots,t\}$
		\State $H_{t}(\bar D) \gets
		\begin{mycases}
		&H_{t-1}(\bar D),
		&&\text{if } t \!\notin\! \bar D
		\\
		&f_t(\bar D)
		&&\text{if } t \!\in\! \bar D
		\end{mycases}
		\;\text{\algorithmicforall\ } \bar D \subset X_{t}
		$
		\Statex \,\qquad\;\;\, where
		$f_t(\bar D) = \sum\nolimits_{i \in X_t} A_{i,t}\, H_{t-1}(\bar D \!\setminus\! \{i,t\})$
		\EndFor
		\State \Return $H_n(X_n)$
		\EndProcedure
	\end{algorithmic}
\end{algorithm}

For each $t \in [n]$, let
\begin{subequations}\label{eq:Xt}
	\begin{align}
	a_{t} &:= \max\{t{-}2w,\, 1\},
	\\
	X_{t} &:= \{a_{t}, a_{t}{+}1, \dots, t\},
	\\
	\Delta_{t} &:= \{1,2,\dots,a_{t}{-}1\}.
	\end{align}
\end{subequations}
Note that $|X_{t}|=t$ for $t \leq 2w$,
and $|X_{t}| = 2w{+}1$ for $t > 2w$.
The table $H_{t}$ is indexed by subsets $\bar{D} \in X_{t}$.
In particular, the table has at most $2^{2w+1}$ entries.
As we explain next, each entry $H_{t}(\bar{D})$ is a subhafnian.
Consider the collection
\begin{align}
\mathcal{S} = \{D \subset [n]: \Delta_{t} \subset D \subset [t] \}.
\end{align}
Note that a set $D \!\in\! \mathcal{S}$ is completely determined by its intersection with~$X_{t}$.
So if we let $\bar{D}:=D \cap X_{t}$, there is a one to one correspondence between $\mathcal{S}$ and the subsets of~$X_{t}$.
The subhafnians that we are interested in are
\begin{align}
H_{t}(\bar{D})
\,:=\,
\mylhaf(D)
\,=\,
\mylhaf(\bar{D} \cup \Delta_{t}),
\quad\text{ for }
\bar{D} \subset X_t.
\end{align}
In particular, the loop hafnian of $A$ is the entry $H_{n}(X_n) = \mylhaf([n])$ of table~$H_n$.

The recursion used in Algorithm~\ref{alg:haf} relies on the next lemma.

\begin{lemma}\label{lem:recursion}
	Let $2\!\leq t\!\leq\!n$ and let $D$ be a subset of $[t]$
	that contains $ \Delta_{t}\!\cup\!\{t\}. $
	Then
	\begin{align}
	\mylhaf(D)
	&= \sum_{i \in X_t}
	A_{i t}\; \mylhaf(D \setminus \{i,t\}).
	\end{align}
\end{lemma}

Before proving the lemma,
let us see that the lemma implies the correctness of the algorithm.
Denoting $\bar{D} := D\cap X_{t}$, the above equation can be rewritten as
\begin{align}\label{eq:recursion}
H_{t}(\bar{D}) &=
\sum_{i \in X_t}
A_{i t} \; H_{t-1}(\bar{D}\setminus\{i,t\})
\end{align}
Each iteration of Algorithm~\ref{alg:haf} uses the above recursion formula.
It follows that the subhafnians $H_t(\bar D) = \lhaf(\bar D \cup \Delta_t)$ are computed correctly.

\begin{proof}[Proof of Lemma~\ref{lem:recursion}]
	Given a matching $\pi' \in \myPM(D {\setminus} \{i,t\})$,
	we can obtain a matching in $\myPM(D)$ by adding the edge $(i,t)$.
	This gives a function:
	\begin{align*}
	\bigcup_{i \in X_t} \myPM(D {\setminus} \{i,t\})
	&\;\to\; \myPM(D),
	\quad
	\pi' \;\mapsto\; \pi' \cup \{(i,t)\}.
	\end{align*}
	Conversely, in any matching $\pi \in \myPM(D)$
	we have that $t$ is connected to a unique~$i$,
	and removing the edge $(i,t)$ gives a matching in $\myPM(D {\setminus} \{i,t\})$.
	Therefore, the function defined  above is a bijection.
	Hence,
	\begin{align}
	\mylhaf(D)
	\,&=\, \sum_{\pi \in \myPM(D)} A(\rho)
	\\
	&=\, \sum_{i \in X_t} \; \sum_{\pi' \in \myPM(D\setminus\{i,t\})} \!\!\!
	A_{i t} \cdot A(\pi')
	\\
	&=\, \sum_{i \in X_t} A_{ij} \cdot \mylhaf(D\setminus\{i,t\}).
	\end{align}
\end{proof}

\begin{proof}[Proof of Theorem~\ref{thm:hafnian-band}]
	We already showed correctness,
	so it remains to estimate the complexity.
	In each iteration we need to compute
	$H_t(\bar D)$ for each $\bar D \subset X_t$.
	Since $|X_t| \leq 2w{+}1$, we need to consider at most $2\cdot 4^w$ subsets.
	The recursion formula~\eqref{eq:recursion} requires $O(|X_t|) \!=\! O(w)$ ring operations.
	Hence, each iteration of the algorithm takes $O(w\, 4^w)$,
	and the overall cost is $O(n\, w\, 4^w)$.
\end{proof}

Recall that the probability of detecting a photon pattern with collisions is given by the loop hafnian of an extended adjacency matrix.
For a GBTS with threshold $c$, we have at most $c$ repetitions for each column and row.
By simply applying Theorem~\ref{thm:hafnian-band}, calculating each probability in our Algorithm~\ref{alg:GBTS} is upper bounded by $O^*(4^{w c})$.

However, such scaling, with an exponential dependence on the number of repetitions, is an overestimate.
On one hand, repetitions increase the matrix size, indeed increasing the cost of calculating its hafnian.
On the other hand, the repetition structure does not carry much new information so that we also expect to see certain cost reduction.
Below we devise a faster algorithm which requires $O^*(n(2c{+}2)^{2w+1})$ steps to calculate the loop hafnian of a banded matrix with at most $c$ repetitions. Since this result is not necessary to prove our main result, we omit the detailed proof here. Interested readers can find the full proof in Appendix~\ref{sec:hafnian-rep}.

\begin{theorem}
	\label{thm:hafnian-rep}
	Let $A$ be a complex-symmetric $n\times n$ matrix  with bandwidth~$w$.
	Let $s = (s_1,\dots,s_n)$ be a vector of positive integers,
	and let $\Arep{s}$ be the symmetric matrix
	obtained from $A$ by repeating the $i$-th row and column $s_i$ times.
	Then we can compute $\lhaf(\Arep{s})$ in time
	${O}^*(n\, (2 c{+}2)^{2w+1})$,
	where $c := \! \max\{s_1,\dots,s_n\}$.
\end{theorem}
\begin{proof}
	See Appendix~\ref{sec:hafnian-rep}.
\end{proof}

\section{Gaussian threshold boson sampling with shallow circuits of local gates}
\label{sec:main}
We are now ready to put together everything we have prepared so far and to prove our main result.
\begin{theorem}
\label{thm:main}Consider an $M$-mode, uniformly lossy GBTS problem with threshold $c>0$. If
the unitary transformation consists of, $D$ layers of commuting local gates, the sampling task can be simulated in running time $T = O^*(M^2 \, \mathrm{poly}(c)\, (2c+2)^{16D})$. Consequently, when the linear-optical circuit is shallow, i.e., when $D=O(\log(M))$, the sampling can be simulated efficiently on a classical computer.
\end{theorem}
\begin{proof}
Recall that at the $k$-th step of Algorithm~\ref{alg:GBTS}, we need to calculate $c$ hafnians (the one correspond to zero photons is trivial to calculate): $\lhaf(\tilde{A}^{(k)}_{s_1,\ldots,s_k})$ for $1\leq s_k\leq c$~. The dependence of the adjacency matrices at step $k$ on the samples from the previous step seems to be complicated to analyze. However, an upper bound on the run time of our classical algorithm, by considering the most costly outcome $s=(c,\ldots,c)$, is straightforward.

From Lemma~\ref{lemma:dragonfly}, we know that the adjacency matrix of the reduced $k$-mode Gaussian state  $A^{(k)}$ is a block banded matrix where the blocks have width at most $4D$. 

Since the loop hafnian is invariant under permutations, invoking Lemma~\ref{lemma:permutation} and Theorem~\ref{thm:hafnian-rep}, we know the cost of finishing step $k$ when Algorithm~\ref{alg:GBTS} is executed is upper bounded by
\begin{align}
    T_k = \sum_{i=0}^c \sbrac{2(k{-}1)c+2i}O^*((2c{+}2)^{16D+1})~.
\end{align}

Summing over $k=1,\ldots,M$ we have the following total run time
\begin{align}
    T = O^*(M^2\, \mathrm{poly}(c)\, (2c{+}2)^{16D})~.
\end{align}
For a shallow circuit, we have $D=O(\log(M))$, which gives $T=\mathrm{poly}(M)$ as $c$ is constant. This concludes our proof.
\end{proof}

\section{Conclusion}
\label{sec:conclusion}

We have proved that sampling from Gaussian states prepared by shallow quantum-optical circuits with local interactions can be simulated efficiently. We have introduced Gaussian Boson Threshold Sampling (GBTS), by which we obtain not only a well-defined model ready for mathematical analysis, but also a  faithful representation of real experiments where photon detectors can be overloaded. We have investigated the structure of the adjacency matrices describing Gaussian states prepared by shallow circuits: they are, up to permutations, banded matrices with bandwidth that is proportional to circuit depth. We have also introduced a dynamic-programming algorithm to exploit the structure of banded matrices with repeated rows and columns, a subroutine to be called in our sampling algorithm. Putting together these steps, we obtain an efficient sampling algorithm for shallow circuits with local gates if the depth grows as a logarithm of the number of modes.

Similar results should be obtainable for regular Boson Sampling by using the faster algorithm for calculating permanents of banded matrices~\cite{cifuentes2016efficient}. Actually, we expect the proof for BS can be much simpler than ours since 1) BS is defined with finite photon number and 2) the relevant matrices are simply submatrices of the unitary matrix.
However, a successful proof also requires an algorithm that calculates a polynomial number of probabilities to output one sample. A good candidate is provided by Clifford and Clifford~\cite{clifford2018classical}, though their algorithm requires careful examination to make it compatible with the results from  Ref.~\cite{cifuentes2016efficient}. Such examination seems to be neglected in previous papers~\cite{muraleedharan2019quantum,lundow2019efficient} in which the authors attempted to leverage the idea of banded matrices.

The complexity of our algorithm agrees with the intuition that weakly-entangled multipartite states are not hard for a classical computer to simulate. Our results rule out optical systems with a shallow-depth circuit and local gates as a candidate for the demonstration of quantum supremacy.
Moreover, if the photon loss compounds exponentially with the circuit depth, which is the case for most of today's photonic platforms, it was shown that the sampling becomes asymptotically simulable~\cite{garcia2019simulating,qi2020regimes}. This poses a great challenge to the demonstration of quantum supremacy using optical circuits with local interactions. Therefore, our work calls for further study of the computational complexity of sampling photonic states generated by shallow-circuits with \emph{non-local} gates~\cite{lubasch2018tensor}.

\section*{Acknowledgements}
H.Q. thanks Ra\'ul Garc\'ia-Patr\'on S\'anchez for  suggesting the initial idea of this project and Daniel Brod and Alexander M.~Dalzell for helpful discussion.
\appendix
\section{Correctness of Algorithm~\ref{alg:GBTS}}
\label{sec:append-correctness-algorithm}
Here we show that Algorithm~\ref{alg:GBTS} indeed samples from the distribution defined in Eq.~\ref{eq:GBTS-dist}. When the algorithm does not halt, the sampled pattern does not overload any PNR, that is $\bm{n}\notin\Sigma_c$ and its probability is given by
\begin{align}
    \tilde{p}(\bm{n}) = \prod_{k=1}^{M}p(n_k\vert n_1\ldots n_{k-1}) = p(\bm{n})~.
\end{align}
When the algorithm does halt, the probability that it halts at step $k$ is given by
\begin{align}
    p(\text{halt at kth step}) = \sum_{n_1,\ldots, n_{k-1}<c, n_k\geq c}\prod_{l=1}^kp(n_k\vert n_1\ldots n_{k-1})~.
\end{align}
Then the total probability of the algorithm halted, output $\#$, is given by
\begin{align}
    p(\#)=&\sum_{k=1}^M p(\text{halt at }k\text{th step})\nonumber\\
     =& \pbrac{\sum_{n_1\geq c, n_2,\ldots,n_M}+ \sum_{n_1\leq c, n_2\geq c,n_3\ldots,n_M} + \ldots} p(\bm{n})\nonumber\\
    =&\sum_{\bm{n}\in\Sigma_c}p(\bm{n})~.
\end{align}
Therefore, indeed our Algorithm~\ref{alg:GBTS} simulates GBTS.

\section{Banded hafnian example calculation}
\label{apendix-example}

Here we illustrate  our loop hafnian algorithm for banded matrices by considering an adjacency matrix with five vertices and having bandwidth equal to one ($w = 1$), namely
\begin{align}
\lhaf\left( \left[\begin{array}{ccccc}
0 & a & 0 & 0 & 0 \\
a & 0 & b & 0 & 0 \\
0 & b & c & d & 0 \\
0 & 0 & d & 0 & e \\
0 & 0 & 0 & e & f \\
\end{array}\right] \right)
\;=\; a c e + a d f.
\end{align}
We denote by $\mylhaf(D)$ the loop hafnian of the submatrix of $A$ given by indices in~$D$.
Trivially, $\lhaf(A) = \mylhaf(\{1,2,\ldots,n\}) $.
The algorithm calculates many such subhafnians,
storing their values in the dynamic programming table.
By reusing previous subhafnians
as well as omitting several subsets $D$ based on the bandwidth,
the algorithm is able to calculate the hafnian of the overall matrix in ${O}(n\, w\, 4^{w})$ arithmetic operations.

The algorithm proceeds in $n$~steps.
In the $t$-th step it calculates the subhafinans given by the subsets $D$
such that $\{1,2,\ldots,t{-}2w{-}1\} \subseteq D \subseteq \{1,2,\ldots,t\}$.
These subhafnians are computed using the formula:
\begin{align*}
\mylhaf(D) = \sum_{i=t-w}^{t} A_{i t}\, \mylhaf(D\setminus \{i,t\}).
\end{align*}

\begin{description}
	\item[Step 1]
	$D \subseteq \{1\}$
	
	Compute $\mylhaf(\emptyset) \!=\! 1$ and $\mylhaf(\{1\}) \!=\! 0$.
	
	\item[Step 2]
	$D \subseteq \{1, 2\}$
	
	Besides the subhafnians from the previous step,
	compute $\mylhaf(\{2\}) \!=\! 0$ and $\mylhaf(\{1,2\}) \!=\! a$.
	
	\item[Step 3]
	$D \subseteq \{1, 2, 3\}$
	
	Four new subhafnians:
	$\mylhaf(\{3\}) \!=\! c$, $\mylhaf(\{1,3\}) \!=\! 0$, $\mylhaf(\{2,3\}) \!=\! b$, and $\mylhaf(\{1,2,3\}) \!=\! a c$.
	
	\item[Step 4]
	$\{1\} \subseteq D \subseteq \{1, 2, 3, 4\}$
	
	Four new hafnians
	($D \!=\! \{1,4\}$, $\{1,2,4\}$, $\{1,3,4\}$, $\{1,2,3,4\}$).
	In particular, $\mylhaf(\{1,2,3,4\})$ is obtained with the formula:
	\begin{align*}
	A_{3,4} \cdot \mylhaf&(\{1,2\}) + A_{4,4} \cdot \mylhaf(\{1,2,3\})
	\\
	&=\, d \cdot a + 0 \cdot a c
	\,=\, a d,
	\end{align*}
	where the needed subhafnians were already computed.
	
	\item[Step 5]
	$\{1,2\} \subseteq D \subseteq \{1, 2, 3, 4,5\}$
	
	Four new hafnians.
	In particular,
	$\mylhaf(\{1,2,3,4,5\})$ is obtained with the formula:
	\begin{align*}
	A_{4,5} \cdot \mylhaf&(\{1,2,3\}) + A_{5,5} \cdot \mylhaf(\{1,2,3,4\})
	\\
	&=\, e \cdot a c + f \cdot a d
	\,=\, a c e + a d f,
	\end{align*}
	where the needed subhafnians were known.
	The above value is the loop hafnian of the original matrix.
	
\end{description}

\section{Loop hafnian algorithm for matrices with repetitions}
\label{sec:hafnian-rep}
In this section we derive an efficient algorithm for computing loop hafnians of banded matrices with repeated entries.
Let $A$ be a symmetric $n\times n$ matrix with bandwidth~$w$
and let $s = (s_1,\dots,s_n) \in \NN^n$ be a vector of positive integers.
We will compute the loop hafnian of the matrix $\Arep{s}$,
obtained by repeating the $i$-th row and column $s_i$ times.

The special case $n=1$ is quite important for the analysis.
In this case the matrix $A$ and the vector $s$ are scalars,
so $\Arep{s}$ is a constant matrix.
Let $T_k(a)$ denote the loop hafnian of an $k\times k$ constant matrix
with all entries equal to~$a$.
The sequence $\{T_k(a)\}_{k \in \NN}$ satisfies the following recursion:
\begin{gather*}
T_0(a) = 1,
\qquad
T_1(a) = a,
\\
T_{k}(a) = a\,\bigl(T_{k-1}(a) + (k{-}1)T_{k-2}(a)\bigr).
\end{gather*}
In particular, $T_k(1)$ is the $k$-th telephone number \cite{bjorklund2019faster}.

Consider now an arbitrary~$n$.
Let $G=(V,E)$ be the underlying graph structure of~$A$,
with vertex set $V = [n]$.
We will define some generalized perfect matchings of~$G$ that allow repeated edges.
We represent a list of repeated edges as a vector $\tau \!\in\! \NN^E$,
i.e., $\tau$ is a vector indexed by~$E$,
and for each $ij \!\in\! E$ the entry $\tau_{ij} \!\in\! \NN$ indicates the number of times that edge $ij$ appears.
The \emph{degree vector} of $\tau$
is the vector $\deg(\tau) \in \NN^n$
with coordinates $\deg(\tau)_{i} := \sum_{j \in [n]} \tau_{ij}$.
For a weight vector $s \in \NN^n$,
we define an \emph{$s$-matching} of $G$ as a vector $\tau \in \NN^E$
such that $\deg(\tau) = s$.
Let $\myPM(s)$ be the set of all $s$-matchings of~$G$.

The following lemma expresses $\lhaf(\Arep{s})$ as sum of some simple quantities associated to each matching $\tau \in \myPM(s)$.
The contribution of $\tau$ to $\lhaf(\Arep{s})$ is a multiple of
\begin{align}
A(\tau) \,:=\,
\left(\prod_{ii \in E_\ell} T_{\tau_{ii}}(A_{ii}) \right)\,
\left(\prod_{ij \in E_0} (A_{ij})^{\tau_{ij}}\right),
\end{align}
where $E_\ell \subset E$ consists of the loops in the graph,
and $E_0 := E \setminus E_\ell$ consists of the remaining edges.

\begin{lemma}\label{thm:hafrepeated}
	Let $A$ be a symmetric $n\times n$ matrix,
	and let $s \in \NN^n$.
	Then
	\begin{align}
	\lhaf(\Arep{s}) \,=\,
	s! \sum_{\tau \in \myPM(s)} \frac{1}{\tau!} \cdot A(\tau),
	\end{align}
	where $ s! := \prod_{i \in [n]} s_{i}!$
	and $ \tau! \,:=\, \prod_{ij \in E} \tau_{ij}!\,. $
\end{lemma}
\begin{proof}
	Assume first that the graph has no loops ($E_\ell \!=\! \emptyset$).
	Let $G_s$ be the graph associated to~$\Arep{s}$.
	We may view its vertices as pairs $(i,\ell_{i})$ where $i \in [n]$ and $\ell_{i} \in [s_{i}]$,
	and its edges do not depend on~$\ell_{i}$.
	Given a perfect matching $\pi \in \PM(G_s)$
	there is a natural way to obtain an $s$-matching $\tau \in \myPM(s)$.
	Namely, for each edge $(i,\ell_{i}), (j,\ell_{j})$ in $\pi$ we ignore the second coordinate and obtain the edge $(i,j)$.
	This gives a function $f : \myPM(s)$.
	It is clear that $\Arep{s}(\pi) = A(f(\pi))$.
	Given $\tau \in \myPM(s)$,
	a simple combinatorial argument shows that the fiber $f^{-1}(\tau)$ consists of exactly $s!/\tau!$ elements.
	Then
	\begin{align*}
	\lhaf(\Arep{s})
	= \sum_{\pi \in \PM(G_{s})} \Arep{s}(\pi)
	= \sum_{\tau \in \myPM(s)} \frac{s!}{\tau!}\cdot A(\tau).
	\end{align*}
	Consider now the general case, where loops are allowed.
	In such a case the equation $\Arep{s}(\pi) = A(f(\pi))$ is no longer valid,
	but we still have that
	\begin{align*}
	\sum_{\pi \in f^{-1}(\tau)} \Arep{s}(\pi)
	= \frac{s!}{\tau!}\cdot A(\tau).
	\end{align*}
	Hence, the argument from above still applies.
\end{proof}

In order to compute $\lhaf(\Arep{s})$,
we will use a variant of the dynamic program in Algorithm~\ref{alg:haf}.
One of the main differences is that the dynamic programming table $H_t$
is now indexed by vectors $\bar{d} \in \NN^n$ instead of subsets $\bar{D} \subset [n]$.

We need additional notation to present the algorithm.
From now on we assume that the matrix $A$
and the weight vector~$s$ (entrywise positive) are fixed.
Given another weight vector $d \in \NN^n$, the associated \emph{scaled subhafnian} is
\begin{align}
\mylhaf(d) \,:=\, \frac{1}{d!} \cdot \lhaf(\Arep{d})
\,=\, \sum_{\tau \in \myPM(d)} \frac{1}{\tau!} \cdot A(\tau).
\end{align}
The support of $d$ is $\supp(d) := \{i \!\in\! [n]: d_{i} \!\neq\! 0\}$.
Given $X \!\subset\! [n]$, the \emph{restriction} $d|_{X} \!\in\! \NN^n$ is obtained by setting the entries outside of $X$ to zero, i.e.,
\begin{align}
(d|_{X})_{i} =
\begin{cases}
d_{i}, &\text{ if }i \in X\\
0, &\text{ if }i \in [n]\setminus X
\end{cases}
\end{align}
We only consider weight vectors $d$ such that $d \leq s$ entrywise.
The \emph{saturated indices} of $d$ are
$\sat(d) := \{i\in [n] : d_{i} \!=\! s_{i}\}$.
Note that $\sat(d) \subset \supp(d)$.
Finally, we denote $\NN_{\leq d}^{X}$ the set of weight vectors supported by $X$ and upper bounded by $d$:
\begin{align}
\NN_{\leq d}^{X} :=
\{ e \in \NN^n \,:\,
\supp(e) \subset X, \;\; e \leq d
\}.
\end{align}

\begin{algorithm}
	\caption{Banded Loop Hafnian with Repetitions}
	\label{alg:hafrep}
	\begin{algorithmic}[1]
		\Procedure{LHafBandRep}{$A,w,s$}
		\State $X_1 := \{1\}$
		\State $H_1(\bar{d}) :=
		\begin{mycases}
		&1,
		&&\text{if } \bar{d} \!=\! 0
		\\[-3pt]
		&\tfrac{1}{\bar{d_1}!} \, T_{\bar{d}_1}(A_{11}),
		&&\text{if } \bar{d} \!\neq\! 0
		\end{mycases}
		\;\text{\algorithmicforall\ } \bar{d} \in \NN_{\leq s}^{X_1}
		$
		\For {$t = 2,\dots,n$}
		\State $a_{t} := \max\{t{-}2w, 1\}$,\;\; $X_t := \{a_{t},\dots,t\}$
		\State $\tilde H_{t-1\!}(\bar{d}) :=
		\begin{mycases}
		&0,
		&&\text{if } \bar{d}_{t} {>} 0
		\\[-3pt]
		&h(\bar{d}),\!\!
		&&\text{if } \bar{d}_{t} {=} 0
		\end{mycases}
		\;\text{\algorithmicforall\ } \bar{d} \in \NN_{\leq s}^{X_t}
		$
		\Statex \,\qquad\;\;\, where
		$h(\bar{d}) = H_{t-1}(\bar{d} + s|_{\{a_{t}{-}1\}})$
		\State $G_{t}(\bar{d}) :=
		\begin{mycases}
		&0,
		&\text{if } \bar{d} \!\notin\!\! \mathcal{D}
		\\[-3pt]
		&g(\bar{d}),
		&\text{if } \bar{d} \!\in\!\! \mathcal{D}
		\end{mycases}
		\;\text{\algorithmicforall\ } \bar{d} \in \NN_{\leq s}^{X_t}
		$
		\Statex \,\qquad\;\;\, where
		$Y_t := X_t \setminus \{t\},$\;\;
		$\sigma_t(\bar{d}) := \bar{d}_t - \sum_{i \in Y_t} \bar{d}_i,$
		\Statex \,\qquad\;\;\,
		$g(\bar{d}) =
		\tfrac{1}{\sigma_t(\bar{d})!} \, T_{\sigma_t(\bar{d})}(A_{t t}) \cdot
		\textstyle\prod_{i \in Y_t} \tfrac{1}{\bar{d}_i!}\, (A_{i t})^{\bar d_{i}},
		$
		\Statex \,\qquad\;\;\, and
		$\mathcal{D} := \{ \bar{d} :
		\sigma_t(\bar{d}) \geq 0\}$
		\State $H_{t} := $\Call{Convolution}{$G_{t}, \tilde H_{t-1},X_t$}
		\EndFor
		\State \Return $s! \cdot H_n(s|_{X_n})$
		\EndProcedure
		\Procedure{Convolution}{$H',H'',X$}
		\State $H(\bar{d}) := \displaystyle\sum\limits_
		{\substack{
				\bar{d}' + \bar{d}''= \bar{d}}}
		H'(\bar{d}') H''(\bar{d}'')$
		\;\; \algorithmicforall\ $ \bar{d} \in \NN_{\leq s}^{X}$
		\EndProcedure
	\end{algorithmic}
\end{algorithm}

Our dynamic program to compute $\lhaf(\Arep{s})$ is given in Algorithm~\ref{alg:hafrep}.
For each $t \in [n]$ we compute a table $H_{t}$.
The loop hafnian of $\Arep{s}$ can be obtained from the last table~$H_n$.
For each $t \in [n]$, let $a_t, X_t, \Delta_t$ be as in~\eqref{eq:Xt}.
The table $H_{t}$ is indexed by vectors $\bar{d} \in \NN_{\leq s}^{X_t}$.
Note that the number of such vectors is
$\prod_{i \in X_t} (s_{i}{+}1) \leq (c{+}1)^{2w+1}$,
where $c := \! \max\{s_1,\dots,s_n\}$.
The entries $H_{t}(\bar{d})$ of the table are scaled subhafnians.
Consider the collection
\begin{align}
\mathcal{S} = \{d \in \NN^n:
\Delta_{t} \subset \sat(d)
\subset \supp(d) \subset [t]
\}.
\end{align}
Observe that $d \!\in\! \mathcal{S}$ is completely determined by the restriction to~$X_t$.
So if we let $\bar{d}:=d|_{X_t}$, there is a one to one correspondence between $\mathcal{S}$ and $\NN_{\leq s}^{X_t}$.
The scaled subhafnians that we are interested in are
\begin{align*}
H_{t}(\bar{d}) :=
\mylhaf(d) =
\mylhaf(\bar{d} + s|_{ \Delta_{t}}),
\quad\text{ for }
\bar{d} \in \NN_{\leq s}^{X_t}.
\end{align*}
In particular, $H_{T}(s|_{X_n}) = \mylhaf(s) = \frac{1}{s!}\lhaf(\Arep{s})$.

The recursion used in Algorithm~\ref{alg:hafrep} relies on the next lemma.

\begin{lemma}\label{thm:recursion2}
	Let $t>1$ and let $d \in \NN^n$ be such that
	\begin{align*}
	\Delta_{t} \subset \sat(d)
	\subset \supp(d) \subset [t],
	\qquad
	d \leq s.
	\end{align*}
	Denoting
	$Y_t := X_t \!\setminus\! \{t\}$
	and $\sigma_t(d') := d'_t -\sum_{i \in Y_t} d'_i$,
	let
	\begin{gather*}
	g(d') :=
	\textstyle\frac{1}{\sigma_t(d')!} \, T_{\sigma_t(d')}(A_{tt}) \cdot
	\prod\nolimits_{i \in Y_t} \frac{1}{d'_i!}\, (A_{it})^{d_{i}'},
	\\
	\mathcal{D} := \{ d' \!\in\! \NN_{\leq d}^{X_t} :
	d'_{t} \!=\! d_{t}, \sigma_t(d')\!\geq\! 0 \}.
	\end{gather*}
	Then
	\begin{gather*}
	\mylhaf(d)
	= \sum_{d' \in \mathcal{D}}
	g(d')\; \mylhaf(d {-} d'),
	\\
	\Delta_{t-1} \subset \sat(d{-}d')
	\subset \supp(d{-}d') \subset [t{-}1]
	\;\;\;\forall d' \!\in\! \mathcal{D}.
	\end{gather*}
\end{lemma}

\begin{proof}[Proof of Lemma~\ref{thm:recursion2}]
	We start with the second equation.
	Let $d' \!\in\! \mathcal{D}$ and $d'':= d{-}d'$.
	Since $d'_{t} \!=\! d_{t}$ then $d''_{t} \!=\! 0$,
	and hence $\supp(d'') \!\subset\! [t{-}1]$.
	If $j \!\in\! \Delta_{t-1} \!\subset\! \Delta_{t}$ then $j \!\notin\! X_t$,
	and hence $d'_{j} \!=\! 0$, $d''_{j} \!=\! d_{j} \!\neq\! 0$.
	Therefore, $\Delta_{t-1} \subset \sat(d'')$.
	
	We proceed to the first equation.
	The proof is quite similar to that of Lemma~\ref{lem:recursion}.
	Given $d' \!\in\! \mathcal{D}$,
	consider the matching
	$\tau(d') \in \myPM(d')$
	defined as follows:
	\begin{align*}
	\tau(d')_{ij} :=
	\begin{cases}
	\sigma_t(d'), &\text{if }i \!=\! j \!=\! t ,\\
	d'_{i}, &\text{if }i \!\in\! Y_t,\, j \!=\! t, \\
	0, & \text{otherwise.}
	\end{cases}
	\end{align*}
	Note that
	$g(d') = A(\tau(d')) / (\tau(d')!) $.
	Consider the function
	\begin{align*}
	\bigcup_{d' \in \mathcal{D}} \myPM(d {-} d') &\;\to\; \myPM(d),
	\qquad
	\tau'' \;\mapsto\; \tau(d') + \tau''.
	\end{align*}
	It can be shown that this function is a bijection.
	Then
	\begin{align*}
	&\mylhaf(d)
	\,=\, \sum_{\tau \in \myPM(d)} \frac{1}{\tau!} \cdot A(\tau).
	\\
	&=\, \sum_{d' \in \mathcal{D}} \; \sum_{\tau'' \in \myPM(d-d')}
	\frac{1}{\tau(d')! \cdot \tau''!} \cdot A(\tau(d')) \cdot A(\tau'').
	\\
	&=\, \sum_{d' \in \mathcal{D}} g(d') \cdot \mylhaf(d\!-\!d') .
	\qedhere
	\end{align*}
\end{proof}

\begin{proof}[Proof of Theorem~\ref{thm:hafnian-rep}]
	We first prove correctness.
	The recursion formula in Lemma~\ref{thm:recursion2} is a \emph{convolution}
	over the vectors ${d} \!\in\! \NN^n$ supported on~$X_t$.
	Each iteration of Algorithm~\ref{alg:hafrep} performs such a convolution.
	It follows by induction on $t$ that the values $H_{t}(\bar{d})$ computed by Algorithm~\ref{alg:hafrep} are indeed given by $\mylhaf(\bar{d} + s|_{\Delta_{t}})$,
	so the algorithm is correct.
	
	We proceed to estimate the complexity.
	We only analyze the cost of the convolution,
	since this is the dominant term.
	Recall that the fast Fourier transform allows us to compute circular convolutions.
	We can avoid the circular effect by appending some zeros.
	It follows that the (non-circular) convolution can be computed in
	$ {O}\left( w \, (2 c {+}2)^{2w{+}1} \log c \right) $,
	where $c := \! \max\{s_1,\dots,s_n\}$.
	The running time of the whole algorithm is
	$ {O}\left( n\, w \, (2 c {+} 2)^{2w{+}1} \log c \right) $.
\end{proof}

The complexity in Theorem~\ref{thm:hafnian-rep} depends on the largest entry of~$s$.
This is not convenient if, for instance, there is a single large entry.
Denoting $S_t := \prod_{i=t}^{t+2 w} (2 s_{i} {+} 2)$
and $C := \! \max\{S_1,\dots,S_n\}$,
our bound can be refined to
${O}^*(n\, C)$. An analogous results exists for generic (as opposed to banded) matrices derived by Kan~\cite{kan2008moments}. 

Although Theorem~\ref{thm:hafnian-rep} is stated only for complex matrices,
Algorithm~\ref{alg:hafrep} can be applied in more general rings.
The complexity remains the same as long as the ring admits a fast Fourier transform.

\bibliography{shallowGBS}
\end{document}